\documentclass{article}

\usepackage{graphicx}
\usepackage[utf8]{inputenc}
\usepackage{amsmath, amssymb, amsthm}
\usepackage{gensymb}
\usepackage{csquotes}
\usepackage[english]{babel}
\usepackage{eurosym}
\usepackage{fullpage}
\usepackage{palatino}
\usepackage{verbatim}
\usepackage{fancybox}

\usepackage{subfig}
\usepackage{enumerate}
\usepackage{enumitem}
\usepackage{epsfig}
\usepackage[active]{srcltx}
\usepackage{amsfonts}
\usepackage{amsmath}
\usepackage[mathlines]{lineno}
\usepackage{xspace}
\usepackage{tikz}
\usepackage{authblk}
\usepackage{soul}

\usepackage{algorithmic}
\usepackage{algorithm}
\floatname{algorithm}{Procedure}

\usepackage[textsize=footnotesize,color=green!40]{todonotes}
\definecolor{red}{RGB}{255,0,0}
\definecolor{blue}{RGB}{0,0,255}

\usetikzlibrary{arrows,shapes,snakes,automata,backgrounds,petri,patterns}

\newtheorem{theorem}{Theorem}
\newtheorem{lemma}[theorem]{Lemma}

\newtheorem{claim}{Claim}
\newtheorem{corollary}[theorem]{Corollary}

\newcommand{\utype}{%
	\begin{tikzpicture}[baseline=0pt, scale=0.07]
	\node[fill, circle, inner sep= 1.4pt, outer sep=0pt] (a) at (0,0) {} ;
	\node[inner sep = 0pt] (c) at (0,4) {} ;
	\draw[->, thick] (a) -- (c) ;
	\node[inner sep = 0 pt] at (1.2, 0) {} ;
	\end{tikzpicture}}

\newcommand{\dtype}{%
	\begin{tikzpicture}[baseline=-1ex, scale=0.07]
	\node[fill, circle, inner sep= 1.4pt, outer sep=0pt] (a) at (0,0) {} ;
	\node[inner sep = 0pt] (c) at (0,-4) {} ;
	\draw[->, thick] (a) -- (c) ;
	\node[inner sep = 0 pt] at (1.2, 0) {} ;
	\end{tikzpicture}}

\newcommand{\etype}{%
	\begin{tikzpicture}[baseline=-0.5ex, scale=0.07]
	\node[fill, circle, inner sep= 1.4pt, outer sep=0pt] (a) at (0,0) {} ;
	\node[inner sep = 0 pt] at (1.2, 0) {} ;
	\end{tikzpicture}}

\title{Computing maximum cliques in $B_2$-EPG graphs}

\author{Nicolas Bousquet\thanks{Supported by ANR Projects STINT
		(\textsc{anr-13-bs02-0007}) and LabEx PERSYVAL-Lab
		(\textsc{anr-11-labx-0025-01}). } 
	\authorcr G-SCOP (CNRS, Univ. Grenoble-Alpes), Grenoble, France. 
	\and Marc Heinrich\thanks{Supported by the \textsc{anr-14-ce25-0006} project of the French National Research Agency}
	\authorcr LIRIS (Université Lyon 1, CNRS), Lyon, France, UMR5205. }

\begin{document}

\maketitle

\begin{abstract}
EPG graphs, introduced by Golumbic et al. in 2009, are edge-intersection graphs of paths on an orthogonal grid. The class $B_k$-EPG is the subclass of EPG graphs where the path on the grid associated to each vertex has at most $k$ bends.
Epstein et al. showed in 2013 that computing a maximum clique in $B_1$-EPG graphs is polynomial. As remarked in~[Heldt et al., 2014], when the number of bends is at least $4$, the class contains $2$-interval graphs for which computing a maximum clique is an NP-hard problem. The complexity status of the Maximum Clique problem remains open for $B_2$ and $B_3$-EPG graphs.
In this paper, we show that we can compute a maximum clique in polynomial time in $B_2$-EPG graphs given a representation of the graph.

Moreover, we show that a simple counting argument provides a ${2(k+1)}$-approximation for the coloring problem on $B_k$-EPG graphs without knowing the representation of the graph. It generalizes a result of~[Epstein et al, 2013] on $B_1$-EPG graphs (where the representation was needed).
\end{abstract}

\section{Introduction}

An \emph{Edge-intersection graph of Paths on a Grid} (or \emph{EPG graph} for short) is a graph where vertices can be represented as paths on an orthogonal grid, and where there is an edge between two vertices if their respective paths share at least one edge. A turn on a path is called a \emph{bend}. EPG graphs were introduced by  Golumbic, Lipshteyn and Stern in~\cite{GolumbicLS09}. They showed that every graph can be represented as an EPG graph. The number of bends on the representation of each vertex was later improved in~\cite{HeldtKU14}.
EPG graphs have been introduced in the context of circuit layout, which can be modeled as paths on a grid. EPG graphs are related to the knock-knee layout model where two wires may either cross on a grid point or bend at a common point, but are not allowed to share an edge of the grid. 
In~\cite{GolumbicLS09}, the authors introduced a restriction on the number of bends on the path representing each vertex. The \emph{class $B_k$-EPG} is the subclass of EPG graphs where the path representing each vertex has at most $k$ bends. Interval graphs (intersection graphs of intervals on the line) are $B_0$-EPG graphs.
The class of trees is in $B_1$-EPG~\cite{GolumbicLS09}, outerplanar graphs are in $B_2$-EPG~\cite{HeldtKU142} and planar graphs are in $B_4$-EPG~\cite{HeldtKU142}. Several papers are devoted to prove structural and algorithmic properties of EPG-graphs with a small number of bends, see for instance~\cite{AlconBDGmRV16,AsinowskiR12,CohenGR14,FrancisL16}.

While recognizing and finding a representation of a graph in $B_0$-EPG (interval graph) can be done in polynomial time, it is NP-complete to decide if a graph belongs to $B_1$-EPG~\cite{CameronCH16} or to $B_2$-EPG~\cite{PergelR16}. The complexity status remains open for more bends. Consequently, in all our results we will mention whether the representation of the graph is needed or not.

Epstein et al.~\cite{EpsteinGM13} showed that the $k$-coloring problem and the $k$-independent set problem are NP-complete restricted to $B_1$-EPG graphs even if the representation of the graph is provided.
Moreover they gave $4$-approximation algorithms for both problems when the representation of the graph is given.
Bessy et al.~\cite{BougeretBGP15} proved that this there is no PTAS for the $k$-independent set problem on $B_1$-EPG graphs and that the problem is $W[1]$-hard on $B_2$-EPG graphs (parameterized by $k$). Recently, Bonomo et al.~\cite{BonomoPS16} showed that every $B_1$-EPG graph admits a $4$-clique coloring and provides a linear time algorithm that finds it, given the representation of the graph.

\paragraph{Maximum Clique problem on EPG graphs.}

\begin{figure}
	\centering
	\begin{tikzpicture}
		\foreach \x in {0,1,2} {
			\draw[line width=1pt] (0,0.1*\x) -| ++(1+0.5*\x,2) -- ++(3,0);
			\draw[line width=1pt] (1.1+0.5*\x,0.1*\x) -| ++(3,2) -- ++(2-0.5*\x,0);
		}
	\end{tikzpicture}
	
	\caption{A complete graph $K_{6}$ minus a matching.}
	\label{fig:bipcomp}
\end{figure}
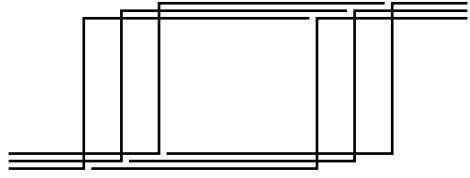

A \emph{claw} of the grid is a set of three edges of the grid incident to the same point. Golumbic et al proved in~\cite{GolumbicLS09} that a maximum clique in a $B_1$-EPG graph can be computed in polynomial time if the representation of the graph is given.
This algorithm is based on the fact that, for every clique $X$ of a $B_1$-EPG graph, either there exists an edge $e$ of the grid such that all the vertices of $X$ contain $e$, or there exists a claw $T$ such that all the vertices of $X$ contain at least two of the three edges of $T$. In particular, it implies that the number of maximal cliques in $B_1$-EPG graphs is polynomial. Epstein et al.~\cite{EpsteinGM13} remarked that the representation of the graph is not needed since the neighborhood of every vertex is a weakly chordal graph.
When the number of bends is at least $2$, such a proof scheme cannot hold since there might be an exponential number of maximal cliques. Indeed, one can construct a complete graph minus a matching in $B_2$-EPG (see Figure~\ref{fig:bipcomp}) which has $2^{n/2}$ maximal cliques. So to compute a maximum clique on $B_{k}$-EPG graphs for $k \geq 2$, a new proof technique has to be introduced.

EPG graphs are closely related to two well known classes of intersection graphs, namely $k$-interval graphs, and $k$-track interval graphs on which the maximum clique problem have been widely studied. A \emph{$k$-interval} is the union of $k$ distinct intervals in the real line. A \emph{$k$-interval graph}, introduced in~\cite{TrotterH79}, is the intersection graph of $k$-intervals. A \emph{$k$-track interval} is the union of $k$ intervals on $k$-distinct lines (called tracks). A \emph{$k$-track interval graph} is an intersection graph of $k$-track intervals (in other words, it is the edge union of $k$ interval graphs on distinct lines). One can easily check, as first observed in~\cite{HeldtKU14}, that $B_{3k-3}$-EPG graphs contain $k$-track interval graphs and $B_{4k-4}$-EPG graphs contain $k$-interval graphs.

Since computing a maximum clique in a $2$-interval graph is NP-hard~\cite{FrancisGO2015}, the Maximum Clique Problem is NP-hard on $B_4$-EPG graphs. So the complexity status of the Maximum Clique problem remains open on $B_k$-EPG graphs for $k=2$ and $3$.
In this paper, we prove that the Maximum Clique problem can be decided in polynomial time on $B_2$-EPG graphs when the representation of the graph is given.
The proof scheme of~\cite{GolumbicLS09} cannot be extended to $B_2$-EPG graphs. Indeed, there cannot exist a bijection between local structures, like claws, and maximal cliques since there are examples with an exponential number of different maximum cliques. Our proof is divided into two main lemmas. The first one ensures that we can separate so-called Z-vertices (vertices that use paths on two rows) from U-vertices (vertices that use edges of two columns). The second ensures that if a graph only contains Z-vertices, then all the maximal cliques are included in a polynomial number of subgraphs; subgraphs for which a maximum clique can be computed in polynomial time.

\paragraph{Coloring $B_k$-EPG graphs.}
Gyarf\'as proved in~\cite{Gyarfas85} that the chromatic number of $k$-interval graphs is bounded by a function of the maximum clique using the degeneracy of the graph. In Section~\ref{sec:coloring}, we propose a slightly different proof than the one of Gyarf\'as to prove the degeneracy of $B_k$-EPG graphs. This bounds ensures that there is a polynomial time algorithm that colors the graph with $2(k+1) \cdot \chi(G)$ colors in polynomial time without knowing the representation of the graph, where $\chi(G)$ is the chromatic number of $G$. In particular, it provides a simple coloring algorithm using at most $4$ times the optimal number of colors on $B_1$-EPG graphs without knowing its representation. It improves the algorithm of~\cite{EpsteinGM13} where the representation was needed.

A class of graphs $\mathcal{C}$ is \emph{$\chi$-bounded} if there exists a function $f$ such that $\chi(G) \leq f(\omega(G))$ for every graph $G$ of $\mathcal{C}$ with $\omega(G)$ the size of a maximum clique in $G$. Combinatorial bounds on the chromatic number of generalizations of interval graphs received a considerable attention  initiated in~\cite{GyarfasL85}. The bound on the degeneracy of the graph ensures that the class $B_k$-EPG is $\chi$-bounded and
${\chi(G) \leq 2(k+1) \cdot \omega(G)}$.
As a by-product, it also ensures that graphs in $B_k$-EPG contain either a clique or a stable set of size $\sqrt{\frac{n}{2(k+1)}}$ which improves a result of~\cite{AsinowskiR12} in which they show that every $B_1$-EPG graph admits a clique or a stable set of size $\mathcal{O}(n^{1/3})$.

\section{Preliminaries}

Let $a,b$ be two real values with $a \leq b$. The \emph{(closed) interval} $[a,b]$ is the set of points between $a$ and $b$ containing both $a$ and $b$. The interval that does not contain $b$ is represented by $[a,b)$ and the one that does not contain $a$ by $(a,b]$.

An \emph{interval graph} is an intersection graph of intervals in the line. More formally, vertices are intervals and two vertices are incident if their respective intervals intersect. Let $H$ be an interval graph with its representation. Let $u \in V(H)$. The \emph{left extremity of $u$} is the leftmost point $p$ of the line such that $u$ contains $p$. The \emph{right extremity of $u$} is the rightmost point $p$ of the line such that $u$ contains $p$.

Let $G$ be an EPG graph with its representation on the grid.
In what follows, we will always denote by roman letters $a,b,\ldots$ the rows of the grid and by greek letters $\alpha,\beta,\ldots$ the columns of the grids. Given a row $a$ (resp. column $\alpha$) of the grid, the row $a-1$ (resp. $\alpha-1$) denotes the row under $a$ (resp. at the left of $\alpha$). Given a row $a$ and a column $\alpha$, we will denote by $(a, \alpha)$ the grid point at the intersection of $a$ and $\alpha$. By abuse of notation, we will also denote by $\alpha$ (for a given row $a$) the point at the intersection of row $a$ and column $\alpha$.
Let $u \in V(G)$. We denote by $P_u$ the path representing $u$ on the grid. The vertex $u$ of $G$ \emph{intersects} a row (resp. column) if $P_u$ contains at least one edge of it.

\section{Typed intervals and projection graphs}
\label{sec:typed}

Let $G$ be a $B_2$-EPG graph with its representation on the grid. Free to slightly modify the representation of $G$, we can assume that the path associated to every vertex has exactly $2$ bends. Indeed, if there is a vertex $u$ such that $P_u$ has less than two bends, let $(a,\alpha)$ be one of the two extremities of $P_u$. Up to a rotation of the grid, we can assume that the unique edge of $P_u$ incident to $(a, \alpha)$ is the horizontal edge $e$ between $(a,\alpha)$ and $(a,\alpha+1)$. Then create a new column $\beta$ between $\alpha$ and $\alpha+1$, and replace the edge $e$ by two edges, one between $\alpha$ and $\beta$ on row $a$, and another one going up at $\beta$. 
This transformation does not modify the graph $G$. So we will assume in the following that for every vertex $u$, the path $P_u$ has exactly two bends.

A \emph{Z-vertex} of $G$ is a vertex that intersects two rows and one column. A \emph{U-vertex} is a vertex that intersects one row and two columns. The \emph{index} of a vertex $u$ is the set of rows intersected by $u$. The vertex $u$ \emph{contains $a$ in its index} if $a$ is in the index of $u$. 

Let $u$ be a vertex containing $a$ in its index. The \emph{extremities} of $u$ on $a$ are the points of the row $a$ on which $P_u$ stops or bends. Since $P_u$ has at most two bends, $P_u$ has exactly two extremities on $a$ and the subpath of $P_u$ on row $a$ is the interval of $a$ between these two extremities. 
The \emph{$a$-interval} of $u$, denoted by $P_u^a$, is the interval between the two extremities of $u$ on $a$. Note that since the index of $u$ contains $a$, $P_u^a$ contains at least one edge.
Let $\alpha \leq \beta$ be two points of $a$. 
The path $P_u^a$ \emph{intersects non-trivially} $[\alpha,\beta]$ if $P_u^a \cap [\alpha,\beta]$ is not empty or reduced to a single point. The path $P_u^a$ \emph{weakly contains $[\alpha,\beta]$} if $[\alpha,\beta]$ is contained in $P_u^a$. 

\subsection{Typed intervals}

Knowing that $P_u^a= [\alpha,\beta]$ is not enough to understand the structure of $P_u$. Indeed whether $P_u$ stops at $\alpha$, or bends (upwards or downwards) at $\alpha$, affects the neighborhood of $u$ in the graph. To catch this difference we introduce typed intervals that contain information on the ``possible bends'' on the extremities of the interval.

We define three \emph{types} namely the \emph{empty type} $\etype$, the \emph{d-type} $\dtype$ and the \emph{u-type} $\utype$. A \emph{typed point} (on row $a$) is a pair $x,\alpha$  where $x$ is a type and $\alpha$ is the point at the intersection of row $a$ and column $\alpha$.
A \emph{typed interval} (for row $a$) is a pair of typed points $(x,\alpha)$ and $(y,\beta)$ (on row $a$) with $\alpha \leq \beta$ denoted by $[x\alpha, y \beta]$.
Informally, a typed interval is an interval $[\alpha, \beta]$ and indications on the structure of the bends on the extremities. A typed interval $t$ is \emph{proper} if $\alpha \neq \beta$ or if $\alpha=\beta$, $x=y$, and $x \in \{ \utype, \dtype \}$.

 Let $t=[x\alpha,y\beta]$ and $t' = [x' \alpha', y' \beta']$ be two typed intervals on a row $a$. Denote by $z\gamma$ one of the endpoints of $t'$. We say that $t$ is \emph{coherent} with the endpoint $z\gamma$ of $t'$ if one of the following holds: (i) $\gamma$ is included in the open interval $(\alpha, \beta)$, or (ii) $z = \etype$, and $[\alpha, \beta]$ contains the edge of $[\alpha', \beta']$ adjacent to $\gamma$, or (iii) $z \neq \etype$, and $z \gamma \in \{ x\alpha, y \beta\}$.
 We can remark in particular that if $t$ is coherent with an endpoint $z\gamma$, then $\gamma$ is in the closed interval $[\alpha, \beta]$.
 The typed interval $t$ \emph{contains} $t'$ if $[\alpha',\beta'] \subset [\alpha, \beta]$, and $t$ is coherent with both endpoints of $t'$.
 The typed interval $t$ \emph{intersects} $t'$ if $[\alpha,\beta]$ intersects non trivially $[\alpha', \beta']$ (i.e. the intersection contains at least one grid-edge), or $t$ is coherent with an endpoint of $t'$, or $t'$ is coherent with an endpoint of $t$. Note that, if $t'$ contains $t$ then in particular it intersects $t$.

Let $u$ be a vertex containing $a$ in its index. The \emph{t-projection of $u$ on $a$} is the typed interval $[x\alpha,y\beta]$ where $\alpha,\beta$ are the extremities of $P_u^a$ and the type of an extremity $\gamma \in \{ \alpha,\beta\}$ is $\etype$ if $P_u$ stops at $\gamma$ and $\dtype$ (resp. $\utype$) if $P_u^a$ bends downwards (resp. upwards) at $\gamma$. Note that this typed interval is proper since it contains at least one edge.
The path $P_u$ \emph{contains} (resp. \emph{intersects}) a typed interval $t$ (on $a$) if the t-projection of $u$ on $a$ contains $t$ (resp. intersects $t$). By abuse of notation we say that $u$ contains or intersects $t$. Note that by definition, the path $P_u$ contains the t-projection of $u$ on $a$.
Moreover if a vertex $u$ contains a typed interval $t = [x\alpha, y\beta]$, then the path $P_u$ weakly contains $[\alpha, \beta]$. If $u$ intersects $t$, then the path $P_u$ intersects the segment $[\alpha, \beta]$ (possibly on a single point).

\begin{figure}
	\centering
	\includegraphics[width=.5\textwidth]{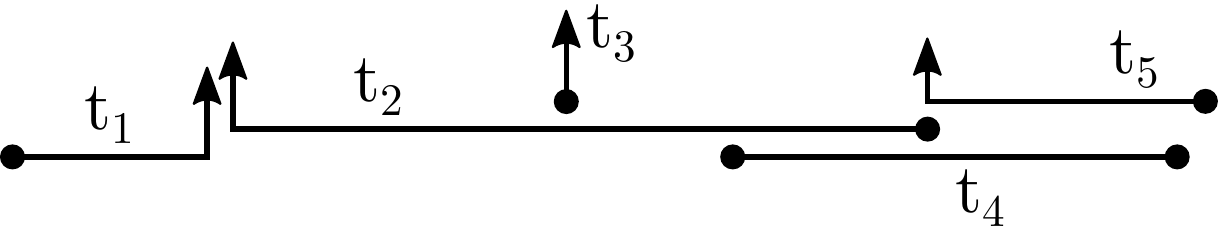}
	\caption{\label{fig:exampleIntersection} Examples of typed intervals on the same row. In this example, the interval $t_3$ is reduced to a single point. The interval $t_2$ is coherent with the right extremity of $t_1$, the extremities of $t_3$ and the left extremity of $t_4$. It is not coherent with the extremities of $t_5$. Moreover, $t_2$ intersects $t_1$, $t_3$ and $t_4$ but not $t_5$. And $t_2$ contains $t_3$.}
\end{figure}

The following simple lemma motivates the introduction of typed intervals.

\begin{lemma}
\label{lem:intersect}
Let $G$ be a $B_2$-EPG graph, let $u, v$ be two vertices whose index contain $a$, and $t$ be a proper typed interval of $a$. If $u$ contains $t$ and $v$ intersects $t$, then $u$ and $v$ are adjacent.
\end{lemma}

\begin{proof}
	Let $t_u = [x\alpha, y \beta]$ and $t_v = [x' \alpha', y' \beta']$ be the t-projections of $u$ and $v$ respectively. 
	Let $t = [x_t \alpha_t, y_t \beta_t]$ be the typed interval such that $u$ contains $t$, and $v$ intersects $t$. Then $[\alpha_t, \beta_t] \subseteq [\alpha, \beta]$, and $[\alpha_t, \beta_t] \cap [\alpha', \beta'] \neq \emptyset$. This implies in particular that the intersection of $[\alpha, \beta]$ and $[\alpha', \beta']$ is not empty.

	If this intersection contains a grid edge $e$, then the paths $P_u$ and $P_v$ both contain this edge, and $u,v$ are adjacent in $G$.

	If the intersection is reduced to a single point, then, since we have $\alpha \neq \beta$ and $\alpha' \neq \beta'$, we can assume w.l.o.g. that the right endpoint of $t_v$ coincide with the left endpoint of $t_u$, i.e. $\beta' = \alpha$. By assumption on $u$, we know that $\alpha \leq \alpha_t$. Additionally, since $[\alpha_t, \beta_t] \cap [\alpha', \beta'] \neq \emptyset$, we must have $\alpha_t = \alpha$. Then either $t_v$ is coherent with the endpoint $x\alpha$ of $t$, or $t$ is coherent with the endpoint $y \beta'$ of $t_v$. We will assume the former, the other case can be treated exactly the same way. Since $[\alpha', \beta']$ does not contain the edge at the left of $\alpha_t$, we have necessarily $y' \neq \etype$. This implies that $y'= x_t$, and $x_t \in \{ \utype, \dtype \}$.
	Additionally we know that $t_u$ is coherent with the endpoint $x_t \alpha_t$ of $t$. Consequently, we have $x = x_t = y'$, which implies that both $P_u$ and $P_v$ bend at $\alpha$ in the same direction, and thus $u,v$ are adjacent in $G$.
\end{proof}

%
%

\begin{lemma}\label{lem:proj}
Let $G$ be a $B_2$-EPG graph. If the t-projections of $u$ and $v$ on $a$ intersect, then $uv$ is an edge of $G$.
Moreover if $u$ and $v$ are two vertices containing $a$ in their index and have no other row in common, then $uv$ is an edge of $G$ if and only if the t-projections of $u$ and $v$ on $a$ intersect.
\end{lemma}

\begin{proof}
	Let $u$ and $v$ be two vertices containing $a$ in their index. The first part of the statement is just a corollary of Lemma~\ref{lem:intersect} since $v$ intersects the t-projection of $u$ on~$a$.

	Assume that there is no other row $b$ contained in the index of both $u$ and $v$. And suppose moreover that the t-projections of $u$ and $v$ on $a$ do not intersect. Suppose by contradiction that $u$ and $v$ are adjacent in $G$. The two vertices cannot share a common edge on row $a$, otherwise their projections would intersect. By assumption they cannot share an edge on another row. Consequently, they must have a common edge $e$ on a column, and let $\alpha$ this column. By symmetry, we can assume that $e$ is below the row $a$. Since the paths $P_u$ and $P_v$ have at most two bends, both path must bend downwards at the intersection of row $a$ and column $\alpha$, to intersect the edge $e$. However, this implies that the t-projections of $u$ and $v$ on $a$ intersect, a contradiction.
\end{proof}

\subsection{Projection graphs}

Let $Y$ be the subset of vertices of $G$ such that all the vertices of $Y$ contain $a$ in their index. \emph{The projection graph of $Y$} on $a$ is the graph on vertex set $Y$ such that there is an edge between $u$ and $v$ if and only if the t-projections of $u$ and $v$ on $a$ intersect.
Note that the projection of $Y$ is not necessarily an interval graph since it can contain an induced $C_4$ (see Fig.~\ref{fig:c4induced}). We say that a set of vertices $Y$ is a \emph{clique on $a$} if the t-projection of $Y$ on $a$ is a clique.
Lemma~\ref{lem:proj}, ensures that a clique on $a$ is indeed a clique in the graph $G$.

\begin{figure}
	\centering
	\includegraphics[width=.4\textwidth]{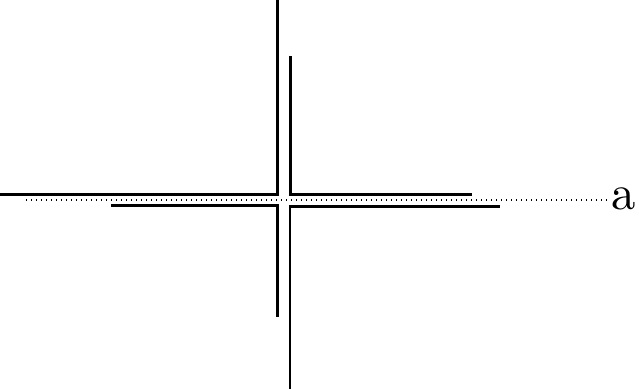}
	\caption{\label{fig:c4induced} In this example, the projection graph on row $a$ of the four vertices is an induced cycle of length four.}
\end{figure}

In the very simple case where the representation uses only two rows $a$ and~$b$, we have the following lemma:
\begin{lemma}
	\label{lem:2track}
Let $G$ be a $B_2$-EPG graph and $G_{ab}$ be the subset of vertices with index $\{a,b\}$. 
Then $G_{ab}$ induces a $2$-track interval graph.
\end{lemma}

\begin{proof}
Let us first prove that the projection graph of $G_{ab}$ on $a$ is an interval graph. Let $u$ be a vertex of $G_{ab}$ and $x$ and $y$ be the endpoints of $P_u^a$. The vertex $u$ bends either on $x$ or on $y$ (since its index is $\{a,b\}$ it has only one vertical segment). We associate to $u$ the segment $s_u = (x,y]$ if $u$ bends on $y$, and $[x,y)$ if it bends on $x$. Then the t-projection of $u$ and $v$ intersect on $a$ if and only if the intervals $s_u$ and $s_v$ intersect. 
Indeed, if both $u$ and $v$ have a bend at the same column $\alpha$, then they both contain the interval between row $a$ and $b$ on column $\alpha$ because they have the same index $\{a,b\}$. Then the projection graph of $G_{ab}$ on $a$ is an interval graph. 

By Lemma~\ref{lem:proj}, any edge on the projection on $a$ or $b$ is also an edge of $G$. Conversely, if there is an edge between $u$ and $v$ then: 
\begin{itemize}
	\item either their $a$-segments (resp. $b$ segments) intersect on at least one edge, and there is an edge in the projection graph of $G_{ab}$ on $a$ (resp. $b$),
	\item or they share the same column. Then since they bend on the same point of the grid and on the same direction, there is an edge between $u$ and $v$ in the projection graphs of $G_{ab}$ on both $a$ and $b$.
\end{itemize}
So there is an edge between $u$ and $v$ if and only if their t-intervals intersect on $a$ or intersect on $b$. Since the t-projections of $G_{ab}$ on both $a$ and $b$ induce interval graphs, the graph $G_{ab}$ is a $2$-track interval graph.
\end{proof}

Let us end this section with two lemmas that will be widely used all along the paper.

\begin{lemma}
	\label{lem:cliqueProj}
	Let $G$ be a $B_2$-EPG graph and $Y$ be a subset of vertices whose index contain $a$. Suppose that the projection of $Y$ on $a$ is a clique. Then there is a proper typed interval $t$ such that:
	\begin{itemize}
		\item all vertices of $Y$ contain $t$,
		\item if $u$ is a vertex with index $\{a\}$ or $\{a,c\}$ where $c$ is not in the index of any vertex of $Y$, then $u$ is complete to $Y$ if and only if $u$ intersects $t$.
	\end{itemize}
\end{lemma}

\begin{proof}
	Let $\alpha$ be the rightmost left extremity of an $a$-interval of $Y$, and $\beta$ be the leftmost right extremity of an $a$-interval of $Y$. Since the projection of $Y$ on $a$ is a clique, $\alpha \leq \beta$. Let $Y_\alpha$ be the set of vertices of $Y$ whose $a$-segment have left extremity $\alpha$, and $Y_\beta$ be the set of vertices of $Y$ whose $a$-segment have right extremity $\beta$. 
	We define the typed interval $t=[x \alpha, y\beta]$  where $x$ is equal to $\utype$ (resp. $\dtype$) if all the vertices of $Y_\alpha$ bend upwards (resp. downwards) at $\alpha$, and is equal to $\etype$ otherwise. Similarly, $y$ is equal to $\utype$ (resp. $\dtype$) if all the vertices of $Y_\beta$ bend upwards (resp. downwards) at $\alpha$, and is equal to $\etype$ otherwise. Let us prove that $t$ satisfies the conclusion of the lemma.

	One can easily check that, by construction, all the vertices of $Y$ contain the typed interval $t$. Indeed $[\alpha,\beta]$ is contained in all the intervals $P_v^a$ for $v \in Y$ by definition of $\alpha$ and $\beta$. 
	Let us prove now by contradiction that the typed interval $t$ is proper. If $\alpha \neq \beta$ then $t$ is proper, so we can suppose that $\alpha=\beta$. Now assume that the types $x$ and $y$ are distinct or equal to $\etype$. Up to symmetry, we can assume that $x \neq \utype$ and $y \neq \dtype$. Consequently, there exists a vertex $v_1 \in Y$ such that $P_{v_1}^a$ has left extremity $\alpha$ and either starts on $\alpha$ or bends downwards on $\alpha$. Similarly, there exists a vertex $v_2 \in Y$ such that $P_{v_2}^a$ has right extremity $\beta=\alpha$ and either ends on $\alpha$ or bends upwards on $\alpha$. But then the t-projection of $v_1$ and $v_2$ do not intersect, a contradiction since the t-projection of $Y$ on $a$ is a clique. 
	\smallskip
	
	Let us finally prove the second point. Suppose that $u$ intersects $t$, then for all $y \in Y$, $y$ contains $t$, and by Lemma~\ref{lem:proj} $u$ and $y$ are adjacent in $G$. Assume now that $u$ does not intersect $t$. Let us prove that there exists $y \in Y$ that is not incident to $u$. 
	Either $P_u^a=[\alpha',\beta']$ does not intersect $[\alpha,\beta]$ or it intersects it on exactly one vertex. We moreover know that if $\alpha=\beta$ then $P_u^a$ does not contain the edge at the left and at the right of $\alpha$ (otherwise $u$ would intersect $t$). So, up to symmetry, we can assume that $\beta' \leq \alpha$. If $\beta' < \alpha$ then let $v$ be a vertex such that $P_v^a$ has leftmost extremity $\alpha$ (such a vertex exists by definition of $\alpha$). The projections of $u$ and $v$ on $a$ do not intersect and Lemma~\ref{lem:proj} ensures that $u$ is not adjacent to $v$.
	
	Assume now that $P_u^a$ intersects $[\alpha,\beta]$ on a single point. Since $P_u^a$ contains at least one edge, this point is either $\alpha$ or $\beta$, $\alpha$ say. Let $x_u$ be the type of $\beta'$. Since $u$ does not intersect $t$, this means that either $x_u \neq x$ or $x_u = x$ and $x_u = \etype$. Up to symmetry, we can assume $x_u \neq \utype$ and $x \neq \dtype$. So there exists $v \in Y$ such that $P_v^a$ has a left extremity at $\alpha$ and $P_v$ either bends upwards at $\alpha$ or has no bend at $\alpha$. Since $x_u \neq \utype$, $P_u^a$ has right extremity $\alpha$ and either bends downwards at $\alpha$ or has no bend at $\alpha$. So the projections of $u$ and $v$ on $a$ does not intersect. By Lemma~\ref{lem:intersect}, $u$ and $v$ are not adjacent in $G$.
\end{proof}

\begin{lemma}
	\label{lem:notCliqueProj}
	Let $G$ be a $B_2$-EPG graph and $Y$ be a subset of vertices with index $\{a,b\}$. Suppose that the projection of $Y$ on $a$ is not clique. Then there is a proper typed interval $t$ such that:
	\begin{itemize}
		\item All vertices of $Y$ intersect $t$,
		\item Let $u$ be a vertex with index $\{a\}$ or $\{a,c\}$ where $c \neq a,b$. Then $u$ is complete to $Y$ if and only if $u$ contains $t$.
	\end{itemize}
\end{lemma}

\begin{proof}
	Let $\beta$ be the rightmost left extremity of an $a$-interval of $Y$, and $\alpha$ be the leftmost right extremity of an $a$-interval of $Y$. Since the projection of $Y$ on $a$ is not a clique, $\alpha \leq \beta$. 
	Let $Y_\alpha$ be the subset of vertices of $Y$ whose $a$-segment have right extremity $\alpha$, and $Y_\beta$ be the subset of vertices of $Y$ with left extremity $\beta$. We consider the typed interval $t$ given by $t=[x\alpha', y\beta']$ where:
	\begin{itemize}
		\item $x$ is equal to $\utype$ (resp. $\dtype$) if all vertices of $Y_\alpha$ bend upwards (resp. downwards) at $\alpha$, and is equal to $\etype$ otherwise, 
		\item similarly, $y$ is equal to $\utype$ (resp. $\dtype$) if all vertices of $Y_\beta$ bend upwards (resp. downwards) at $\alpha$, and is equal to $\etype$ otherwise,
		\item if $x \neq \etype$ then $\alpha' = \alpha$, otherwise $\alpha' = \alpha-1$,
		\item if $y \neq \etype$ then $\beta' = \beta$, otherwise $\beta' = \beta+1$.
	\end{itemize}

	Let us prove that $t$ is a proper typed interval. We can assume w.l.o.g. that $a$ is below $b$. If $\alpha' \neq \beta'$ then $[\alpha',\beta']$ is non reduced to a single point and then $t$ is proper. So we can assume that $\alpha'=\beta'$. In the construction, if one of $x$ or $y$ is equal to $\etype$, then $\alpha' \neq \beta'$ and then $t$ is proper. Additionally, since all vertices have index $\{a,b\}$, and $a$ is below $b$, all vertices bend upwards on row $a$. Consequently, $x$ and $y$ cannot be equal to $\dtype$. If none of them is $\etype$, then they are both $\utype$, and again $t$ is proper.

	We will first prove that every vertex of $Y$ intersects $t$. Let $v$ be a vertex of $Y_{\alpha}$, then either $\alpha' = \alpha -1$ and then $P_v^a$ contains the edge $[\alpha-1, \alpha]$ and thus $v$ intersects $t$, or $\alpha'=\alpha$. In that case, $x$ is equal to $\utype$ (resp. $\dtype$) if all the vertices of $Y_{\alpha}$ bends upwards (resp. downwards) on $\alpha$. But then the t-projection of $v$ has a right endpoint with extremity $\alpha$ and type $\utype$ (resp. $\dtype$), and then by definition, the t-projection of $v$ is coherent with the endpoint $x\alpha$ of $t$, and $v$ intersects $t$. A similar proof holds to show that t-projections of vertices of $Y_\beta$ intersect $t$. Now let $v \in Y \setminus (Y_{\alpha} \cup Y_{\beta})$. The left extremity of $v$ is before $\alpha$ and its right extremity is after $\beta$. Let $\alpha_v$ and $\beta_v$ be the left and right extremities of $P_v^a$. They satisfy $\alpha_v \leq \alpha$ and $\beta_v \leq \beta$. If $\alpha' = \alpha$, then $\alpha$ is contained in the open interval $(\alpha_v, \beta_v)$, and the t-projection of $v$ on $a$ is coherent with the extremity $x\alpha$ of $t$, and by consequence $v$ intersects $t$. If $\alpha' = \alpha -1$, then the t-projection of $v$ weakly contains $[\alpha', \beta']$, and this interval is not reduced to a single point. Consequently $v$ intersects $t$.
	
	\smallskip
	
	Let us now prove the second point. Let $u$ be a vertex such that the index of $u$ contains $a$ and not $b$. If $u$ contains $t$, then for any $y \in Y$, $y$ intersects $t$, and by Lemma~\ref{lem:proj} $u$ and $y$ are adjacent in $G$. 
	
	Assume now that $u$ does not contain $t$. We will show that there is a $y \in Y$ such that $u$ and $y$ are not adjacent in $G$. By Lemma~\ref{lem:intersect}, it suffices to show that there exists a vertex $v$ such that the t-projections of $u$ and $v$ on $a$ do not intersect. First suppose that $P_u$ does not weakly contain $[\alpha', \beta']$. We can assume w.l.o.g. that $P_u$ does not contain $\alpha'$. Let $\alpha_u$ and $\beta_u$ be the left and right endpoints of $P_u^a$. If $\beta_u < \alpha'$, then there is a vertex $v \in Y_\beta$ such that the left endpoint of $P_v^a$ is $\beta$, and $\alpha_u < \alpha' \leq \alpha \leq \beta$. Consequently the t-projection of $u$ and $v$ do not intersect. Otherwise, we can assume $\alpha_u > \alpha'$. If $\alpha' = \alpha$, then there is a vertex $v \in Y_\alpha$ such that the right endpoint of $P_v^a$ is $\alpha$. Then, since $\alpha = \alpha' < \alpha_u$, the t-projection of $u$ and $v$ do not intersect, and by Lemma~\ref{lem:intersect}, $u$ and $v$ do not intersect. If $\alpha' = \alpha -1$, then there is a vertex $v$ such that $\alpha$ is the right endpoint of $P_v^a$, and $P_v$ stops at $\alpha$. Since $\alpha' < \alpha_u$, we have $\alpha \leq \alpha_u$. This implies that $[\alpha_u, \beta_u]$ does not contain the edge at the left of $\alpha$, which implies that the t-projection of $u$ is not coherent with the right extremity of the t-projection of $v$. The same arguments show that the t-projection of $v$ on $a$ is not coherent with the endpoints of the t-projection of $u$, and thus the t-projections of $u$ and $v$ do not intersect.
	
	Suppose now that $[\alpha_u, \beta_u]$ weakly contains $[\alpha', \beta']$. Let $t_u$ be the t-projection of $u$. Since $u$ does not contain $t$, $t_u$ is not coherent with one of its endpoints, say $\alpha'$ up to symmetry. Then we have $\alpha' \in \{\alpha_u, \beta_u\}$. Suppose by contradiction that $x = \etype$. Then we can't have $\alpha' = \beta_u$ since otherwise $t$ would not be proper. So necessarily $\alpha' = \alpha_u$, but this implies that $t_u$ is coherent with the endpoint $\alpha'$ since $[\alpha_u, \beta_u]$ contains the edge at the left of $\alpha'$, a contradiction. Since $a$ is below $b$, we must have $x = \utype$. If $\alpha' = \alpha_u$, then since $t_u$ is not coherent with the endpoint $\utype \alpha'$ of $t$, necessarily the right endpoint of $t_u$ has type $\etype$. Then, if $v \in Y_\alpha$, then by construction of $t$, since $x = \utype$ we have $\alpha' = \alpha$, and $\alpha_u$ is the right endpoint of $P_v^u$. Since $P_u$ does not bend upwards at $\alpha$, the t-projections of $u$ and $v$ do not intersect. If $\alpha' = \alpha_u$, then $\alpha' = \beta'$, and since $t$ is proper, we necessarily have $y = \utype$. A similar argument shows that there exists $v \in Y_\beta$ whose t-projection does not intersect $t_u$.
\end{proof}

\section{Maximum clique in $B_2$-EPG graphs}

\subsection{Graphs with Z-vertices}\label{sec:Zvertices}

We start with the case where the graph only contains Z-vertices. We will show in Section~\ref{sec:gen} that it is possible to treat independently Z-vertices and U-vertices. The remaining part of Section~\ref{sec:Zvertices} is devoted to prove the following theorem.

\begin{theorem}\label{th:Zclique}
	Let $G$ be a $B_2$-EPG graph with a representation containing only Z-vertices. The size of a maximum clique can be computed in polynomial time.
\end{theorem}

Note that, up to rotation of the representation of $G$, this theorem also holds for U-vertices. In other words, the size of a maximum clique can be computed in polynomial time if the graph only contains U-vertices.

The proof of Theorem~\ref{th:Zclique} is divided in three steps. We will first define a notion of good subgraphs of $G$, and prove that:
\begin{itemize}
	\item there is a polynomial number of good subgraphs of $G$ and,
	\item a maximum clique of a good graph can be computed in polynomial time,
	\item and any maximal clique of $G$ is contained in a good subgraph.
\end{itemize}

Recall that a clique is \emph{maximum} if its size is maximum. And it is \emph{maximal} if it is maximal by inclusion.
The first point is an immediate corollary of the definition of good graphs. The proof of the second point consists in decomposing good graphs into sets on which a maximum clique can be computed efficiently. The proof of the third point, the most complicated, will be divided into several lemmas depending on the structure of the maximal clique we are considering.

An induced subgraph $H$ of $G$ is a \emph{good graph} if one of the following holds:
\begin{enumerate}[label=(\Roman*)]
	\item \label{type:case1} there are two rows $a$ and $b$, and two proper typed intervals $t_a$ and $t_b$ on $a$ and $b$ respectively such that $H$ is the subgraph induced by vertices $v$ such that, $v$ contains $t_a$, or $v$ contains $t_b$, or $v$ intersects both $t_a$ and $t_b$,
	\item \label{type:case2} or there are three rows $a$, $b$, and $c$, and three proper typed intervals $t_a$, $t_b$, and $t_c$ on $a, b, c$ respectively such that $H$ is the subgraph induced by vertices $v$ such that, either $v$ contains $t_a$, or $v$ contains $t_b$, or $v$ intersects $t_b$ and contains~$t_c$.
\end{enumerate}

\begin{lemma}
	\label{lem:goodGraphs}
	Let $G$ be a $B_2$-EPG graph.
	There are $\mathcal{O}(n^6)$ good graphs, and a maximum clique of a good graph can be computed in polynomial time.
\end{lemma}
\begin{proof}
	A good graph is defined by two or three typed intervals. At first glance, to define a typed interval we need to choose a row, and then two points on it. So a natural upper bound on the number of typed interval is $O(n^3)$. However, we can make a slightly better evaluation. A point of the grid is \emph{important} if the path of a vertex ends or bends on this point or on a point incident to it. There is a linear number of important points since every path defines a constant number of important points. One can easily notice, given an interval $t_a$ on row $a$, replacing an extremity $r$ of $t_a$ by the important point $s$ that is the closest from $r$ on row $a$ does not modify the set of vertices intersecting or complete to $t_a$. Indeed, no path starts, stops or bends on the interval $[r,s]$ on row $a$.
	So we can assume that all the endpoints of typed intervals are important points. This implies that there are at most $O(n^2)$ typed intervals. As a consequence, there are $O(n^4)$ good graphs of type~\ref{type:case1}, and $O(n^6)$ good graphs of type~\ref{type:case2}.

	Let us now prove that a maximum clique can be computed in polynomial time in a good graph.
	\begin{itemize}
		\item Let $H$ be a good graph of type~\ref{type:case1} defined by two proper typed intervals $t_a$ and $t_b$. Let $H_a$ be the subset of vertices containing $t_a$, and $H_b$ be the set of vertices containing $t_b$, and $H_{ab}$ be the other vertices of $H$ (intersecting both $t_a$ and $t_b$ by definition of $H$). By Lemma~\ref{lem:intersect}, both $H_a$ and $H_b$ are complete to $H_{ab}$ and both $H_a$ and $H_b$ are cliques. Let $H_1$ be the graph induced by $H_a \cup H_b$, and $H_2$ be the graph induced by $H_{ab}$. Then $H$ is the join of $H_1$ and $H_2$ (i.e. there is an edge between any vertex of $H_1$ and any vertex of $H_2$). So a maximum clique of $H$ is the union of a maximum clique of $H_1$ and a maximum clique of $H_2$. Moreover:
		\begin{itemize}
			\item Since both $H_a$ and $H_b$ induce cliques, $H_1$ is the complement of a bipartite graph. Computing a maximum clique in $H_1$ is the same as computing a maximal independent set in $\bar{H_1}$, the complement graph of $H_1$ which is bipartite.
			Additionally, computing a maximum independent set in a bipartite graph can be done in polynomial time. Indeed, a maximum independent set is the complement of a minimum vertex cover. In bipartite graphs, a minimum vertex cover can be computed in polynomial time using for instance Linear Programming.
			\item By Lemma~\ref{lem:2track}, $H_2$ is a $2$-track interval graph, on which a maximum clique can be computed in polynomial time~\cite{FrancisGO2015}.
		\end{itemize}
		Consequently, a maximum clique of both $H_1$ and $H_2$, and then of $H$, can be computed in polynomial time.
		
		\item Let $H$ be a good graph of type~\ref{type:case2} defined by three proper typed intervals $t_a$, $t_b$ and $t_c$. Let $H_a$ and $H_b$ be the set of vertices containing $t_a$ and $t_b$ respectively, and $H_{bc}$ be the set of vertices intersecting $t_b$  and containing $t_c$. By Lemma~\ref{lem:intersect}, $H_a$, $H_b$, and $H_{bc}$ are cliques since for each of these sets there is a proper typed interval contained in every vertex of the set. Moreover $H_{bc}$ is complete to $H_{b}$ since vertices of $H_b$ contain $t_b$ and vertices of $H_{bc}$ intersect it. Let $H_1= H_a$ and $H_2 = H_b \cup H_{ab}$. Both $H_1$ and $H_2$ induce cliques. So $H$ is the complement of a bipartite graph, and a maximum clique of $H$ can be computed in polynomial time.
	\end{itemize}
\end{proof}

The remaining part of Section~\ref{sec:Zvertices} is devoted to prove that any maximal clique of $G$ is contained in a good graph.

\begin{lemma}
	\label{lem:2rows}
	Let $G$ be a $B_2$-EPG graph containing only Z-vertices, and $X$ be a clique of $G$. Assume that there are two rows $a$ and $b$ such that every horizontal segment of $X$ is included in either $a$ or $b$, then $X$ is contained in a good graph.
\end{lemma}
\begin{proof}
	By taking two typed intervals $t_a$ and $t_b$ consisting of the whole rows $a$ and $b$, the clique $X$ is clearly contained in a good graph of type~\ref{type:case1}.
\end{proof}

We say that a set of vertices $X$ \emph{intersects} a column $\alpha$ (resp. a row $a$) if at least one vertex of $X$ intersects the column $\alpha$ (resp. the row $a$). If $X$ is a clique of $G$, and $a$, $b$ two rows of the grid, we denote by $X_{ab}$ the subset of vertices of $X$ with index $\{a,b\}$.

The three following lemmas allow us to prove that all the cliques of a $B_2$-EPG graph are included in a good subgraphs. They use the same kind of techniques. The main idea is, using the tools of Section~\ref{sec:typed}, to find typed intervals which describe well the vertices in a clique $X$.

\begin{lemma}\label{lem:0clique}
	Let $G$ be a $B_2$-EPG graph containing only Z-vertices, and $X$ be a clique of $G$. If there are two rows $a$ and $b$ such that the projection graphs of $X_{ab}$ on $a$ and $b$ are not cliques, then $X$ is included in a good graph.
\end{lemma}
\begin{proof}
	Let $X$ be a clique satisfying this property for rows $a$ and $b$. Let $X_a$ be the set of vertices of $X \setminus X_{ab}$ intersecting row $a$ and not row $b$, and $X_b$ be the set of vertices of $X \setminus X_{ab}$ intersecting row $b$ and not row $a$.
	
	First note that $X = X_a \cup X_b \cup X_{ab}$. Otherwise there would exist a vertex $w$ of type $(c,d)$ such that $\{c,d\} \cap \{ a,b\} = \emptyset$ in $X$. Since $w$ is complete to $X_{ab}$, $w$ would intersect all the vertices of $X_{ab}$ on their vertical part. But the projection graph of $X_{ab}$ on $a$ is not a clique, consequently $X_{ab}$ intersects at least two columns. Thus a vertex of $X_{ab}$ does not intersect the unique column of $w$, a contradiction.

	By Lemma~\ref{lem:notCliqueProj} applied to $X_{ab}$ on both $a$ and $b$, there exist two typed intervals $t_a$ and $t_b$ on rows $a$ and $b$ such that every vertex of $X_{ab}$ intersects both $t_a$ and $t_b$, and such that vertices of $X_a$ contain $t_a$ and vertices of $X_b$ contain $t_b$. So $X$ is included in a good graph of type~\ref{type:case1} defined by $t_a$ on $a$ and $t_b$ on $b$.
\end{proof}

\begin{lemma}\label{lem:1clique}
	Let $G$ be a $B_2$-EPG graph containing only Z-vertices, and $X$ be a clique of $G$. If there are two rows $a$ and $b$ such that the projection graph of $X_{ab}$ on $a$ is not a clique, then there is a good graph containing $X$.
\end{lemma}
\begin{proof}
	Let $X$ be a clique satisfying this property for rows $a$ and $b$. We can assume that the projection graph of $X_{ab}$ on $b$ is a clique since otherwise we can apply Lemma~\ref{lem:0clique}. Let $X_a$ be the set of vertices of $X \setminus X_{ab}$ intersecting row $a$ and not row $b$, and $X_b$ be the set of vertices of $X \setminus X_{ab}$ intersecting row $b$ and not row $a$.
	
	First note that $X = X_a \cup X_b \cup X_{ab}$. Otherwise there would exist a vertex $w$ of index $\{c,d\}$ such that $\{c,d\} \cap \{ a,b\} = \emptyset$ in $X$. Since $w$ is complete to $X_{ab}$, $w$ would intersect all the vertices of $X_{ab}$ on their vertical part. But the projection graph of $X_{ab}$ on $a$ is not a clique, consequently $X_{ab}$ intersects at least two columns. Thus a vertex of $X_{ab}$ does not intersect the unique column of $w$, a contradiction.
	
	Suppose first that the projection graph of $X_b \cup X_{ab}$ on $b$ is a clique. By Lemma~\ref{lem:cliqueProj} applied to $X_b$ on $b$, there exists a proper typed interval $t_b$ on row $b$ such that every vertex of $X_b$ contains $t_b$. By Lemma~\ref{lem:notCliqueProj} applied to $X_{ab}$ on $a$, there exists a proper typed interval $t_a$ on row $a$ such that every vertex of $X_a$ contains $t_a$.
	So all the vertices of $X$ are contained in the good graph of type~\ref{type:case1} defined by $t_a$ and $t_b$.

	Suppose now that the projection graph of $X_b \cup X_{ab}$ on $b$ is not a clique. Then there are two vertices $u,v$ of $X_b$ such that the t-projections of $u$ and $v$ on $b$ do not intersect. By Lemma~\ref{lem:intersect}, $P_u$ and $P_v$ intersect another row $c$, and since the projection graph of $X_{ab}$ on $b$ is a clique, $c \neq a$. Since the projection graph of $X_{bc}$ on $b$ is not a clique, and $c \neq a$, we can assume that the projection graph of $X_{bc}$ on $c$ is a clique since otherwise Lemma~\ref{lem:0clique} can be applied on $X_{bc}$. So by Lemma~\ref{lem:cliqueProj} applied to $X_{bc}$ on $c$, there exists a typed interval $t_c$ contained in every vertex of $X_{bc}$. By Lemma~\ref{lem:notCliqueProj} applied to $X_{bc}$ on $b$ (resp. $X_{ab}$ on $a$), we know that there exists a typed interval $t_b$ (resp. $t_a$) satisfying the two conditions of Lemma~\ref{lem:notCliqueProj}. Now we divide $X_b \cup X_{ab}$ into two subclasses: $X_{bc}$ and $Y_b= (X_b \cup X_{ab}) \setminus X_{bc}$.
	
	We have $X = Y_b \cup X_{bc} \cup X_a$.
	Vertices of $X_a$ contain $t_a$. Vertices of $X_{bc}$ intersect $t_b$ and contain $t_c$. Vertices of $Y_b$ contain $t_b$. This proves that $X$ is included in the good graph of type~\ref{type:case2} defined by $t_a,t_b$ and $t_c$ on respectively rows $a,b,c$.
\end{proof}

We now have all the ingredients we need to prove Theorem~\ref{th:Zclique}.
By Lemma~\ref{lem:1clique}, we can assume that for all $a,b$, the projection graphs of $X_{ab}$ on $a$ and $b$ are cliques. As a by-product, if we denote by $X_a$  the set of vertices containing $a$ in their index, then the projection graph of $X_a$ on $a$ induces a clique for every row $a$. Indeed, suppose by contradiction that there is a row $a$, and two vertices $u,v$ whose index contains $a$, such that the t-projections of $u$ and $v$ on $a$ do not intersect. Then by Lemma~\ref{lem:intersect}, $u$ and $v$ intersect a common row $b \neq a$. But then, the projection graph of $X_{ab}$ on $a$ is not a clique, a contradiction.
Thus by Lemma~\ref{lem:cliqueProj}, for any row $a$ such that $X_a$ is not empty, there exists a proper typed interval $t_a$ such that the t-projection on $a$ of any vertex of $X_a$ contains $t_a$.

Assume that there are two rows $a$ and $b$ such that $X_{ab}$ intersects at least two columns. By Lemma~\ref{lem:2rows}, we can assume that $X$ intersects at least three rows. Let $w \in X \setminus X_{ab}$ be a vertex using a third row. As in the proof of Lemma~\ref{lem:0clique}, since $X_{ab}$ has at least two columns and $w$ is complete to $X_{ab}$, $w$ has index either $\{a,c\}$ or $\{b,c\}$ for $c \neq a,b$. Consequently, $w$ contains either $t_a$ or $t_b$, which proves that $X$ is contained in a good graph of type~\ref{type:case1}.

Assume now that for all pair of rows $(a,b)$, all the vertices of $X_{ab}$ pass through the same column. We can assume that $X$ intersects at least three columns, otherwise we could just rotate the representation of $G$ and apply Lemma~\ref{lem:2rows}. Let $u, v$ and $w$ be three vertices using three different columns. Let ${a,b}$ be the index of $u$. We can assume w.l.o.g. that the index of $v$ is $\{b,c\}$. There are two possible cases:
\begin{itemize}
	\item $w$ has index $\{b,d\}$ for some row $d$. Note that $d \neq a,c$ since otherwise this would contradict the fact that both $X_{ab}$ and $X_{bc}$ intersect only one column. If all vertices in $X$ contain $b$ in their index, then they all contain the typed interval $t_b$, and $X$ is contained in a good graph of type~\ref{type:case1}. Let $z$ be a vertex with $z \in X \setminus X_b$. Up to permutations of rows $a$, $c$ and $d$, we can assume that the row $c$ is below $a$ and $d$ is below $c$. Suppose that $b$ is over $c$. Then, since $z$ must intersect $u,v$ and $w$, necessarily the index of $z$ is of the form $\{s_1,s_2\}$ with $s_1, s_2 \in \{a,b,d\}$. Indeed, $z$ can intersect only one of $u,v,w$ on its vertical part. Additionally, the index of $z$ cannot be $\{c,d\}$ since otherwise $z$ cannot intersect the vertical part of $w$. (see Fig.~\ref{fig:zcliqueproof}). So necessarily the index of $z$ contains $a$. This proves that every vertex of $X$ contains either $t_b$ or $t_a$, and as a consequence $X$ is included in a good graph of type~\ref{type:case1}.
	A similar argument shows that if $b$ is below $c$, then $z$ cannot be of index $\{a,c\}$, and the index of $z$ contains $d$.
	
				\begin{figure}
					\centering
					\begin{tikzpicture}[scale=0.5]
					\draw[thick] (0,0) -- (5,0) -- (5,3) -- node[midway, above] {$u$} (13,3) ; 
					\draw[thick] (2,-0.2) -- (6,-0.2) -- (6,-2) -- node[midway, above] {$v$} (12,-2) ; 
					\draw[thick] (-1,-0.4) -- (4,-0.4) -- (4,-5) -- node[midway, below] {$w$} (10,-5) ; 
					\draw[line width= 2.5pt, gray!50!white] (9,-2.2) -- (5,-2.2) -- node[midway, right, black] {$z$} (5,-4.8) -- (8,-4.8) ;
					\node (ra) at (17,3) {$a$}; 
					\node (rb) at (17,-0.2) {$b$}; 
					\node (rc) at (17,-2) {$c$}; 
					\node (rd) at (17,-5) {$d$}; 
					\draw[dotted] (0,3) -- (ra) ;
					\draw[dotted] (0,-0.2) -- (rb) ;
					\draw[dotted] (0,-2) -- (rc) ;
					\draw[dotted] (0,-5) -- (rd) ;
					\end{tikzpicture}
					\caption{\label{fig:zcliqueproof} A vertex $z$ with index $\{c,d\}$ cannot intersect $u$ since the vertical part of $u$ is not between rows $c$ and $d$.}
				\end{figure}
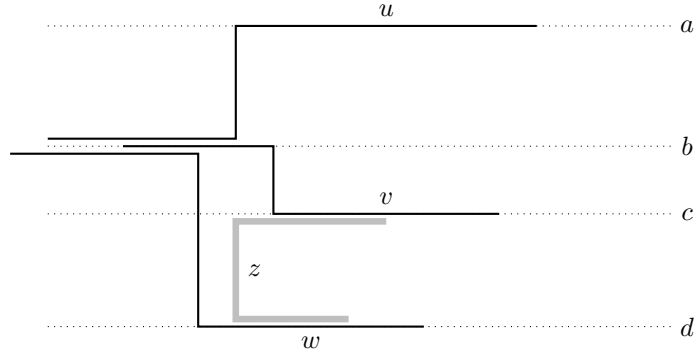

	\item $w$ does not contain $b$ in its index, then $w$ has necessarily index $\{a,c\}$ in order to intersect both $u$ and $v$. If there are two rows, say for example $a$ and $b$, such that every vertex of $X$ has an index containing either $a$ or $b$, then every vertex of $X$ contains either $t_a$ or $t_b$. This implies that $X$ is contained in a good graph of type~\ref{type:case1}. Suppose that there are no two rows such that the index of any vertex contains one of the two rows. Suppose w.l.o.g.. that row $b$ is below $a$, and $c$ below $b$. There is a vertex $z_c$ whose index does not contain $a$ or $b$. Then necessarily, the index of $z_c$ is $\{c, d_c\}$ for a certain row $d_c$. Additionally, $z_c$ must intersect the vertex $u$ of index $\{a,b\}$ on its vertical part. This implies that $b$ is below $d_c$. By a similar argument, there is a vertex $z_a$ of index $\{a,d_a\}$ with $d_a \neq b,c$. Then $z_a$ must intersect $v$ on its vertical part, and $d_a$ is below $b$. In particular, $d_a$ is different from $d_c$. This implies that $z_a$ and $z_c$ do not have a row or a column in common, and thus do not intersect. Since $X$ is a clique, this case is not possible.
	
	\begin{figure}
		\centering
		\begin{tikzpicture}[scale=0.5]
		\draw[thick] (0,0) -- ++ (5,0) -- node[midway, left] {$v$}  ++ (0,3) -- ++ (4,0) ;
		\draw[thick] ++ (6,3.2) -- ++ (2,0) -- node[midway, left] {$u$} ++ (0,2) -- ++ (7,0) ;
		\draw[thick] ++ (1,-0.2) -- ++ (10,0) -- node[midway, right] {$w$} ++ (0,5.2) -- ++ (3,0) ;
		
		\draw[line width=2pt, gray!50!white] ++ (1,-0.4) -- ++ (7.2,0) -- node[midway, right, black] {$z_c$} ++ (0,4) -- ++ (2,0) ;
		\draw[line width=2pt, gray!50!white] ++ (10,5.4) -- ++ (-5.2,0) -- node[midway, right, black] {$z_a$} ++ (0,-3) -- ++ (-2,0) ;
		
		\node (ra) at (20,5.2) {$a$} ;
		\node (rb) at (20,3) {$b$} ;
		\node (rc) at (20,-0.2) {$c$} ;
		
		\node (rda) at (-2,2.4) {$d_a$} ;
		\node (rdc) at (15,3.6) {$d_c$} ;
		
		\draw[dotted] (ra) -- ++(-17,0) ;
		\draw[dotted] (rb) -- ++(-17,0) ;
		\draw[dotted] (rc) -- ++(-17,0) ;
		
		\draw[dotted] (rdc) -- ++(-5,0) ;
		\draw[dotted] (rda) -- ++(5,0) ;
		
		\end{tikzpicture}
		\caption{\label{fig:zcliqueproof2} The vertices $z_a$ and $z_c$ cannot be adjacent since their paths do not have a row or a column in common.}
	\end{figure}
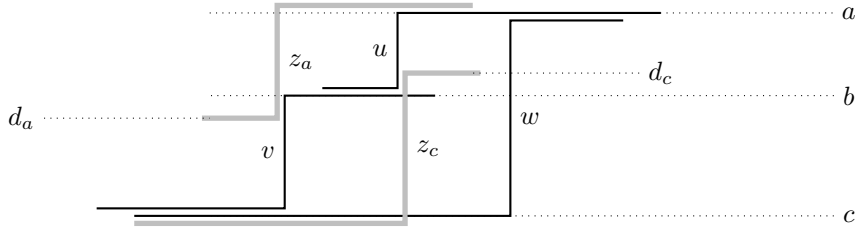

\end{itemize}

\subsection{General $B_2$-EPG graphs}
\label{sec:gen}

In Section~\ref{sec:Zvertices}, we have seen how to compute a maximum clique in a graph containing only Z-vertices. This section is devoted to prove that we can actually separate the graph in order to assume the graph only contains Z-vertices or U-vertices. We start by proving two lemmas showing that the existence of U-vertices puts some constraints on the Z-vertices that can be appear in a clique. We will then use these two lemmas to prove our main theorem.

\begin{lemma}
	\label{lem:atmost3rows}
	Let $G$ be a $B_2$-EPG graph with a representation, and $X$ be a clique of $G$, then:
	\begin{itemize}
		\item either there are $3$ rows intersecting all the U-vertices of $X$
		\item or there are three columns intersecting all the Z-vertices of $X$.
	\end{itemize}
\end{lemma}
\begin{proof}
	Let $u_1, u_2, u_3$, and $u_4$ be four U-vertices of $X$ intersecting pairwise different rows. Let us prove that there are three columns containing every Z-vertex of $X$.
	
	First assume that there are three columns $\alpha$, $\beta$, $\gamma$ such that, the set of columns intersected by $u_i$ is in $\{\alpha,\beta,\gamma\}$ for every $i \leq 4$. Let us prove that these three columns intersect every Z-vertex of $X$. Assume by contradiction that there exists $v$ in $X$ that does not intersect $\alpha$, $\beta$ and $\gamma$. Then for every $i$, $P_v$ shares an edge with $P_{u_i}$ on a horizontal segment. Since all the $u_i$ have disjoint index, this would imply that $v$ intersects four different rows, a contradiction.
	
	So we can assume that $u_1,u_2,u_3,u_4$ intersect at least four columns. Let $\alpha$ and $\beta$ be the columns of $u_1$. We can assume w.l.o.g. that $u_2$ intersects the columns $\alpha$ and $\gamma$, with $\gamma \neq \alpha, \beta$. And that $u_3$ intersects a fourth column $\delta \neq \alpha, \beta, \gamma$. So both $u_3$ and $u_4$ must intersect $\alpha$ since they must intersect both $u_1$ and $u_2$. Let $\tau$ be the second column intersected by $u_4$. Then any Z-vertex of $X$ intersects one of $\alpha, \delta, \tau$. Indeed, suppose by contradiction that a Z-vertex $v$ of $X$ does not intersect one of these columns. Since $P_v$ does not intersect $P_{u_3}$ and $P_{u_4}$ on their vertical intervals, it shares an edge with $P_{u_3}$ and $P_{u_4}$ on their two horizontal parts. Since $u_1,u_2,u_3,u_4$ have pairwise different index, $P_v$ that intersect the row of $u_3$ and the row of $u_4$, share an edge with $P_{u_1}$ on the column $\beta$ and $P_{u_2}$ on the column $\gamma$ since $v$ does not intersect column $\alpha$. However, a Z-vertex intersects a single column, a contradiction.
\end{proof}

In Section~\ref{sec:typed}, we have introduced typed intervals. These typed interval defines intervals on a given row. In the following claim, we need two typed of typed interval: horizontal and vertical typed interval. An \emph{horizontal typed interval} is a typed interval as defined in Section~\ref{sec:typed}. A \emph{vertical typed interval} is a typed interval of the graph after a rotation, i.e. the graph where rows become columns and columns become rows.

\begin{lemma}
	\label{lem:Uintervals}
	Let $G$ be a $B_2$-EPG graph with its representation, and $X$ be a clique of $G$ containing only U-vertices with the same index $\{a\}$. There exists a set $S_t$ of at most three typed intervals such that:
	\begin{itemize}
		\item $S_t$ contains exactly one horizontal typed interval, and at most two vertical typed intervals,
		\item every vertex of $X$ contains all the typed intervals of $S_t$,
		\item a Z-vertex $u$ is complete to $X$ if and only if $u$ intersects one of the typed intervals of $S_t$.
	\end{itemize}
	%
	%
\end{lemma}
\begin{proof}
	Since $X$ is a clique of $G$ and $X$ only contains U-vertices of index $\{a\}$, Lemma~\ref{lem:proj} ensures that the projection graph of $X$ on $a$ is a clique. By Lemma~\ref{lem:cliqueProj} applied to $X$ on $a$, there is a typed interval $t$ such that every vertex of $X$ contains $t$, and, if $u$ is a vertex containing $a$ in its index, and $u$ is complete to $X$, then $u$ must intersect $t$. The typed interval $t$ is the unique horizontal typed interval of~$S_t$.
	
	Suppose that there is a column $\alpha$, such that every vertex of $X$ intersects $\alpha$. Since all the vertices of $X$ intersect the same row $a$ and $X$ is a clique, the projection graph of $X$ on the column $\alpha$ is a clique.
	Indeed, since all the vertices of $X$ intersect column $\alpha$ and row $a$, all of them must bend on the point $(a,\alpha)$. Either they all bend on the same direction on column $\alpha$, say upwards, and then they all contain the edge of the column $\alpha$ between $a$ and $a+1$, and the projection graph is a clique. Or, some vertices of $X$ bend upwards and other downwards on $(a,\alpha)$. Since $X$ is a clique, they all come from the same direction on row $a$ and then their t-projections on $\alpha$ pairwise intersect.
	
	By Lemma~\ref{lem:cliqueProj} applied to $X$ on column $\alpha$, there exists a vertical typed interval $t_\alpha$ satisfying both properties of Lemma~\ref{lem:cliqueProj}.
	Since every U-vertex intersects two columns, there are at most two columns $\alpha,\beta$ for which every vertex of $X$ intersects these columns. Let $S_t$ be the set composed of $t$ and the typed intervals $t_\alpha$ and $t_\beta$ if they exist.
	
	Let us prove that $S_t$ satisfies the conclusion of the lemma. By construction $S_t$ contains exactly one horizontal typed interval and at most two vertical typed intervals. By definitions of $t$ and $t_\alpha, t_\beta$, every vertex $v$ of $X$ contains the typed intervals in $S_t$. Let us finally show the last point.
	Let $u$ be a Z-vertex. If $u$ intersects a typed interval in $S_t$, then by Lemma~\ref{lem:intersect}, $u$ is complete to $X$. Conversely, suppose that $u$ is complete to $X$. If $u$ contains $a$ in its index, then Lemma~\ref{lem:cliqueProj} ensures that $u$ intersects $t$ since vertices of $X$ all have index $\{a\}$. Assume now that the index of $u$ does not contain $a$. So $u$ intersects al the vertices of $X$ on its unique column. Let $\gamma$ be the unique column intersected by $u$. All the vertices of $X$ must intersect $\gamma$ since otherwise $u$ cannot be complete to $X$. Then $\gamma \in \{ \alpha, \beta\}$, and w.l.o.g., we can assume $\gamma = \alpha$. Then Lemma~\ref{lem:cliqueProj} ensures that $u$ intersects $t_\alpha$ since the unique column of $u$ is $\alpha$.
\end{proof}

The two previous lemmas are the main ingredients to prove that a maximum clique in $B_2$-EPG graphs can be computed in polynomial time. The idea of the algorithm is, using Lemma~\ref{lem:Uintervals}, to guess some typed intervals contained in the U-vertices of the clique. Lemma~\ref{lem:atmost3rows} ensures that we do not have to guess too many intervals. Once we have guessed these intervals, we are left with a subgraph which is actually the join of two subgraphs, one with only Z-vertices, and another with only U-vertices. Then the maximum clique is obtained by applying Theorem~\ref{th:Zclique} to each of the components.

\begin{theorem}
	Given a $B_2$-EPG graph $G$ with its representation, there is a polynomial time algorithm computing the maximum clique of $G$.
\end{theorem}
\begin{proof}
	In the rest of this proof, $S_i$ will denote a set of typed intervals. A vertex $u$ \emph{contains} $S_i$ if $u$ contains all the typed intervals of $S_i$. And $u$ \emph{intersects} $S_i$ if $u$ intersects one of its typed intervals. Given $k$ of these sets $S_1, S_2, \ldots S_k$, we denote by $G(S_1, S_2, \ldots S_k)$ the subgraph induced by the set of U-vertices containing one of the $S_i$ and the set of Z-vertices intersecting all of the $S_i$-s.
	
	Let $X$ be a clique of $G$. Let us show that there are at most three sets $S_1,\ldots, S_k$ with $k \leq 3$ such that $X$ is contained in $G(S_1, \ldots, S_k)$. Free to rotate the representation by 90\degree if needed, Lemma~\ref{lem:atmost3rows} ensures that there are at most three rows intersecting all U-vertices. Let us denote by $a_1,a_2,a_3$ these (at most) three rows. By Lemma~\ref{lem:Uintervals} applied on each row $a_i$, there exists a set of typed intervals $S_i$ such that every U-vertex of $X$ intersecting $a_i$ contains $S_i$, and every Z-vertex of $X$ intersects $S_i$. This implies that $X$ is included in $G(S_1, \ldots, S_k)$ with $k \leq 3$.
	
	Let us now describe the algorithm that computes a maximum clique in $G$: guess the sets $S_1, \ldots S_k$ by trying all possibilities, and then compute the maximum clique of $G(S_1, \ldots S_k)$. Reusing the argument in the proof of Lemma~\ref{lem:goodGraphs}, we know that there are at most $O(n^2)$ typed intervals. Since each set $S_i$ is composed of three typed intervals, this give at most $O(n^6)$ possibilities. By looking more precisely at the proof of Lemma~\ref{lem:Uintervals} we can see that the vertical typed intervals share a common endpoint with the horizontal one. This means that actually you only have $O(n^3)$ possibilities to look at for a set $S_i$. Now, since we need at most three of these sets, there are at most $O(n^9)$ possible choices for the sets $S_1, \ldots S_k$, $k \leq 3$.
	
	To complete the proof of the theorem, we only need to prove that computing a maximum clique in a subgraph $G(S_1, \ldots, S_k)$ can be done in polynomial time. Fix $S_1, \ldots S_k$ sets of typed intervals, and denote $H = G(S_1, \ldots, S_k)$. Let $H_Z$ (resp. $H_U$) be the subgraph of $H$ induced by the Z-vertices (resp. U-vertices). Then $H$ is the join of $H_U$ and $H_Z$. Indeed, let $u$ a U-vertex of $H$, and $v$ a Z-vertex of $H$. By construction, there is a set $S_i$ such that $u$ contains $S_i$. Since $v$ intersects $S_i$, by Lemma~\ref{lem:intersect} $u$ and $v$ are adjacent in $H$. By Theorem~\ref{th:Zclique}, a maximum clique of $H_U$ and a maximum clique of $H_Z$ can be computed in polynomial time. This implies that a maximum clique of $H$ can be computed in polynomial time.
\end{proof}

\section{Colorings and $\chi$-boundedness}\label{sec:coloring}

We denote by $\chi(G)$ the \emph{chromatic number} of $G$, i.e. the minimum number of colors needed to properly color the graph $G$. And we denote by $\omega(G)$ the maximum size of a clique of $G$. The following lemma upper bounds the number of edges in a $B_k$-EPG.
A similar bound on the number of edges was proposed by Gyarf\'as in ~\cite{Gyarfas85} for $k$-interval graphs. We nevertheless give the proof for completeness.

\begin{lemma}\label{lem:edges}
Let $G$ be a $B_k$-EPG graph on $n$ vertices. There are at most $(k+1)(\omega(G) - 1)n$ edges in $G$.
\end{lemma}
\begin{proof}
Let $G$ be a $B_k$-EPG graph, and consider a representation of $G$. Let $q$ be the maximum number of distinct paths going through one edge of the grid. Then $q \leq \omega(G)$ since all the paths sharing a common edge of the grid forms a clique of $G$. Let us prove that $G$ has at most $(k+1)(q-1)$ edges. 
	
The path $P_u$ of $u$ can be decomposed into at most $(k+1)$ intervals where an interval is a maximum subpath of $P_u$ between two bends or between the beginning or the end of $P_u$ and the first or last bend. Each interval is then contained on a single row or on a single column. Note that two intervals of $P_u$ can be included on the same row if $k$ is large enough. If two intervals of $P_u$ on a row intersect, we define replace both intervals with a single interval that is the union of the two intervals. This operation is called a merging. 
Given a vertex $u$, the canonical intervals of $u$ are the intervals of $u$ where we merged all the intervals that have an edge in common.
	
Let $G_i=(V_i,E_i)$ be the interval graph where the vertices of $G_i$ are the canonical intervals included in the $i$th row of the representation of $G$. And there is an edge between two canonical intervals if they intersect non trivially. Note that each path $P_u$ might contribute to several vertices in $G_i$ but all these intervals have to be disjoint intervals are canonical.
Note that the graph $G_i$ defines an interval graph. Similarly we define $G'_j$ for each column of the representation. We denote by $G_s$ the graph defined as the disjoint union of all the $G_i$s and of all the $G_j'$s. We claim the following:
	
\begin{claim}
	$|E(G_s)| \geq |E(G)|$ and $\omega(G_s)=q$.
\end{claim}
\begin{proof}
Let $uv$ be an edge of $G$. Then $P_u$ and $P_v$ share an edge and then there is a canonical interval $i$ of $P_u$ and a canonical interval $j$ of $P_v$ such that $i$ and $j$ intersect. We can associate to $uv$ the edge $ij$. This function from $E(G)$ to $E(G_s)$ is clearly injective and then $|E(G_s)| \geq |E(G)|$ .
Assume by contradiction that there is a clique of size more than $q$ in $G_i$. Then there exists an edge of row $i$ containing at least $q+1$ canonical intervals. Since there are at most $q$ different vertices of $G$ going through the same edge, two intervals are associated to the same vertex, a contradiction with the fact that the intervals are canonical.
\end{proof}
The remaining of the proof consists in evaluating the number of edges in $G'$. Now the graphs $G_i$ and $G'_i$ are interval graphs, consequently their number of edges satisfies: $|E(G_i)| \leq (q-1) |V(G_i)|.$ Since each path is composed of at most $(k+1)$ intervals, we get:
	
\begin{align*}
|E(G)| &\leq \sum_i |E(G_i)| + \sum_j |E(G'_j)| \\
&\leq (q-1) \left(\sum_i |V(G_i)| + \sum_j |V(G'_j)|\right) \\
&= (q-1) (k+1) n 
\end{align*}
	
\end{proof}

A graph is \emph{$k$-degenerate} if there is an ordering $v_1,\ldots,v_n$ of the vertices such that for every $i$, $|N(v_i) \cap \{ v_{i+1},\ldots,v_n \}| \leq k$. It is straightforward to see that $k$-degenerate implies $(k+1)$-colorable. Lemma~\ref{lem:edges} immediately implies the following:

\begin{corollary}
 Let $G$ be a $B_k$-EPG graph:
 \begin{itemize}
 \item The graph $G$ is $\Big(2(k+1) \omega-1\Big)$-degenerate.
  \item $ \chi(G) \leq 2(k+1) \omega(G)$.
  \item There is a polynomial time $2(k+1)$-approximation algorithm for the coloring problem without knowing the representation of $G$.
  \item Every graph of $B_k$-EPG contains a clique or a stable set of size at least $\sqrt{\frac{n}{2(k+1)}}$.
 \end{itemize}
\end{corollary}
\begin{proof}
	The first three points are immediate corollaries of Lemma~\ref{lem:edges}. Let us only prove the last point. If the graph does not admit a clique of size at least $\sqrt{\frac{n}{2(k+1)}}$, then 
	\[\chi(G) \leq 2(k+1) \cdot \sqrt{\frac{n}{2(k+1)}} \leq \sqrt{2(k+1)n}.\]
	Since a proper coloring is a partition into stable sets, there exists a stable set of size at least $\frac{n}{\sqrt{2(k+1)n}}= \sqrt{\frac{n}{2(k+1)}}$ which concludes the proof.
\end{proof}

\bibliographystyle{plain}

\bibliography{biblio}

\end{document}